\documentclass[conference]{IEEEtran}
\usepackage[utf8]{inputenc}
\usepackage[english]{babel}
\usepackage[T1]{fontenc}
\usepackage[hidelinks]{hyperref}
\usepackage{amsmath,amsfonts,amsthm,amssymb}
\usepackage{stmaryrd}
\usepackage{graphicx}
\usepackage[all]{xy}
\usepackage{mathpartir}
\usepackage{macros}

\theoremstyle{definition}
\newtheorem{definition}{Definition}

\theoremstyle{plain}
\newtheorem{proposition}[definition]{Proposition}
\newtheorem{theorem}[definition]{Theorem}
\newtheorem{lemma}[definition]{Lemma}

\newtheorem{conjecture}[definition]{Conjecture}
\theoremstyle{remark}
\newtheorem{remark}[definition]{Remark}
\newtheorem{example}[definition]{Example}

\begin{document}
\title{A Type-Theoretical Definition\\of Weak $\omega$-Categories}
\author{
\IEEEauthorblockN{Eric Finster\\Samuel Mimram}
\IEEEauthorblockA{École Polytechnique\\
91192 Palaiseau France}
}


\maketitle

\begin{abstract}
  We introduce a dependent type theory whose models are weak
$\omega$-categories, generalizing Brunerie's definition of
$\omega$-groupoids. Our type theory is based on the definition of
$\omega$-categories given by Maltsiniotis, himself inspired by
Grothendieck's approach to the definition of $\omega$-groupoids. In
this setup, $\omega$-categories are defined as presheaves preserving
globular colimits over a certain category, called a \emph{coherator}.
The coherator encodes all operations required to be present in an
$\omega$-category: both the compositions of \emph{pasting schemes} as
well as their coherences. Our main contribution is to provide a
canonical type-theoretical characterization of pasting schemes as
contexts which can be derived from inference rules. Finally, we
present an implementation of a corresponding proof system.
\end{abstract}

\section{Introduction}
\subsection{Weak $\omega$-categories}
In a strict $\omega$-category, the axioms are designed to ensure that
the composite of any collection of composable cells is uniquely
defined.  Whichever way we choose to compute this composite will always
give rise to the same result. For instance, if we consider the
situation where we have three sequentially composable cells, this
forces composition to be associative. In a weak $\omega$\nbd-category,
our goal is to achieve a similar uniqueness, but without resorting to
equality: two compositions of $n$-cells should be related by an
$(n+1)$-cell (which should be unique up to $(n+2)$-cells, etc.). We
now have coherence cells, which themselves should have coherence
cells, etc.

Achieving a reasonable definition of weak $\omega$-categories is not
an easy task. Many proposals now exist, each with its own geometric
flavor and collection of techniques.  Often these techniques pass
through sophisticated categorical machinery, and making practical use
of the definition can be challenging.  In this paper, we take up the
definition proposed by Grothendieck~\cite{grothendieck1983pursuing}
for $\omega$-groupoids (categories in which every cell is invertible),
and later simplified and extended to a definition of
$\omega$-categories by
Maltsiniotis~\cite{maltsiniotis2010grothendieck}.  This definition was
studied in detail in Ara's thesis~\cite{ara:these} (who showed that it
is equivalent to Batanin's definition using contractible
operads~\cite{batanin1998monoidal}). The first difficulty overcome by
this proposal is the definition of what it means for a collection of
cells to be ``composable'' via the introduction of what we will refer
to as \emph{pasting schemes} in what follows. From here, the
definition mainly consists in formally iteratively adding composites
for such pasting schemes while preserving previously defined
compositions (although there are, of course, some subtleties
here). Note that contrary to the usual, explicit definitions of
low-dimensional weak $n$\nbd-categories (\eg bicategories or
tricategories) which insist on having compositions generated by binary
and nullary (identity) compositions, this definition is ``unbiased''
in the sense that compositions of all reasonable shapes are taken as
primitive operations.

\subsection{A type-theoretical definition}
The goal of this article is to reformulate this definition in
type-theoretic terms, which is to say to present a type theory such
that the (set-theoretic) models of the theory should be precisely weak
$\omega$\nbd-categories. The idea of formulating $1$\nbd-categories in
type theory dates back to
Cartmell~\cite{cartmell1986generalised}. More recently, in his
thesis~\cite{brunerie:these}, Brunerie has introduced a
type-theoretical definition of weak $\omega$-groupoids, with the aim
of showing that types in homotopy type theory possess such a structure
(see also~\cite{lumsdaine2009weak,van2011types,altenkirch2012syntactical}
for other work in this direction). In this article, we generalize and extend his work in
order to give a definition for $\omega$-categories. The main
contribution here is to characterize pasting schemes in type theory,
a step which is not required for defining $\omega$\nbd-groupoids.


There are a number of reasons why one might seek such a
reformulation. First, it provides us with a syntax for
$\omega$\nbd-categories which can be quite convenient in practice: in
particular, one can give meta-theoretic proofs by induction on the
structure of terms. Second, it has didactic merits: our definition
consists in only a few inference rules, and should be comprehensible
to anyone with some experience in logic or type systems. That is, we
keep the categorical prerequisites to a minimum.  Third, it is compact
lending itself to concrete computations. Finally, it is mechanizable
meaning that one can give a typechecking algorithm for determining if
a given term is a valid coherence in an $\omega$-category. To our
knowledge there are only two such tools for checking proofs in higher
categories. The first one is Opetopic~\cite{opetopic}, based on
opetopic categories, coming with a very different definition and tools.
The second one is Globular~\cite{vicary2016globular}, based on the
theory of semi-strict categories. While this theory allows for much
shorter proofs and has a very nice graphical interface for
constructing them, a complete set of axioms which should be satisfied
in high-dimensions is not known yet (not even whether there is a
reasonable such set of axioms); on the other hand, our tool is
based on a firm theory, but requires significantly more small-step
manipulations in the proofs.

\subsection{Plan of the paper}
We begin by introducing a type-theoretical definition of globular sets
(\secr{globular-sets}), then characterize and study pasting schemes among them
(\secr{pasting-schemes}) and use those to define weak $\omega$-categories
(\secr{cat}). We finally briefly present an implementation
(\secr{implementation}) and conclude (\secr{concl}).






The authors would like to thank Dimitri Ara for his helpful
discussions on the topic of this paper. This work was supported by the
CATHRE ANR grant ANR-13-BS02-0005-02.


\section{A type theory for globular sets}
\label{sec:globular-sets}
Before proceeding to the complete definition of $\omega$\nbd-categories,
we introduce first in this section a type theory whose models are
precisely globular sets, see also~\cite{lumsdaine2009weak,van2011types}.
This simpler theory contains only the
context and type formation rules, but we present it here and study
it in detail in order to make our work easier when considering the
complete system in \secr{cat}.

\subsection{Globular sets}
The definition of $\omega$-categories which concerns us here is based on the
notion of globular set. A globular set may be seen as an higher-dimensional
generalization of a (directed) graph, consisting not only of edges, but
of edges between edges and so on.

\begin{definition}
  \label{def:gset}
  A \emph{globular set}~$G$ consists of a family $(G_n)_{n\in\N}$ of sets
  together with two families of maps $s_n,t_n:G_{n+1}\to G_n$ indexed by
  $n\in\N$ such that
  \begin{equation}
    \label{eq:gset}
    s_n\circ s_{n+1}=s_n\circ t_{n+1}
    \qquad\text{and}\qquad
    t_n\circ s_{n+1}=t_n\circ t_{n+1}    
  \end{equation}
  for every $n\in\N$. A \emph{morphism} $f:G\to G'$ between globular sets $G$
  and $G'$ consists of a family of functions $f_n:G_n\to G'_n$ such that
  $s_n\circ f_{n+1}=f_n\circ s_n$ and $t_n\circ f_{n+1}=f_n\circ t_n$ for every
  $n\in\N$. We write $\GSet$ for the resulting category.
\end{definition} 

\noindent
In a globular set~$G$, the elements of $G_n$ are called \emph{$n$-cells} (cells
whose \emph{dimension} is $n$) and the functions~$s_n$ and~$t_n$ respectively
associate to an $(n+1)$-cell its \emph{source} and \emph{target} $n$-cell. We
say that a cell is \emph{top-dimensional} when it is neither the source nor the
target of another cell.
We sometimes write $s^m_n=s_n\circ s_{n+1}\circ\ldots\circ s_{m-1}$ for the
\emph{iterated source} function.  The \emph{iterated target} function, $t^m_n$,
is defined similarly.

The set of all cells of a globular set~$G$ is denoted
$G_\infty=\coprod_{n\in\N}G_n$ and the \emph{cardinal} of~$G$ is that
of~$G_\infty$. We say that~$G$ is \emph{finite} when $G_\infty$ is finite, or
equivalently when~$G_n$ is finite for any $n\in\N$ and there is~$N\in\N$ such
that $G_n=\emptyset$ for every $n\geq N$. The full subcategory on finite
globular sets is denoted~$\FinGSet$.

\begin{example}
  \label{ex:gset}
  The diagram
  $
  \vxym{
    x\ar@/^1.5ex/[r]^f\ar@/_1.5ex/[r]_g\ar@{}[r]|{\phantom\alpha\Downarrow\alpha}&y&\ar[l]_hz
  }
  $
  depicts the globular set~$G$ with $G_0=\set{x,y,z}$, $G_1=\set{f,g,h}$,
  $G_2=\set{\alpha}$ and $G_n=\emptyset$ for $n\geq 3$, with $s_1(\alpha)=f$,
  $t_1(\alpha)=g$, $s_0(f)=x$, $t_0(f)=y$, etc.
\end{example}

Equivalently, the category~$\GSet$ of globular sets can be defined as the
category $\hat\Glob$ of presheaves over the category~$\Glob$ whose objects are
integers and morphisms are generated by $s_n,t_n:n\to n+1$, for $n\in\N$,
subject to relations which are dual of~\eqref{eq:gset}. As with any presheaf
category, we are provided with the Yoneda embedding $Y:\Glob\to\hat\Glob$. Given
an object $n\in\Glob$, we write $D_n=Yn$ and call it the \emph{$n$-disk}: its
set of $k$-cells is $\set{x_k^-,x_k^+}$ for $k<n$, $\set{x_k}$ for $k=n$ and
$\emptyset$ otherwise:
\[
\begin{array}{c@{\hspace{3ex}}c@{\hspace{3ex}}c@{\hspace{3ex}}c}
  \xymatrix{
    x_0
  }
  &
  \xymatrix{
    x_0^-\ar[r]^{x_1}&x_0^+
  }
  &
  \xymatrix{
    x_0^-\ar@/^/[r]^{x_1^+}\ar@/_/[r]_{x_1^-}\ar@{}[r]|{\phantom{x_2}\Downarrow x_2}&x_0^+
  }
  &
  \xymatrix{
    x_0^-\ar@/^/[r]^{x_1^+}\ar@/_/[r]_{x_1^-}\ar@{}[r]|{\Downarrow\TO\Downarrow}&x_0^+
  }
  \\
  D_0&D_1&D_2&D_3
\end{array}
\]
We also write $\sigma_n^m=Ys_n^m:D_n\to D_m$ (\resp $\tau_n^m=Yt_n^m$) for the
canonical inclusion of an $n$-disk as the source (\resp target) of an $m$-disk.

\noindent
Equivalently, globular sets can also be defined coinductively:

\begin{definition}
  \label{def:gset-coind}
  A \emph{globular set}~$G$ consists of a set~$G$ together with, for all
  elements $x,y\in G$, a globular set~$G_{x,y}$.
\end{definition}

\subsection{Syntactic constructions}
We suppose fixed an infinite countable set of \emph{variables} $x,y,\ldots$. A
\emph{term} in the theory will always be a variable in this section. (The
distinction between terms and variable will become meaningful starting from
\secr{cat}). A \emph{substitution}~$\sigma$ is a list
\[
\sigma
\qeq
\slist{t_1,\ldots,t_n}
\]
of terms $t_i$, the empty substitution being denoted $\sempty$. The
\emph{types} are defined inductively as being either
\[
\Obj
\qquad\text{or}\qquad
\Hom Atu
\]
where $A$ is a type and $t$ and $u$ are terms. A \emph{context} $\Gamma$ is a
list
\[
\Gamma
\qeq
x_1:A_1,\ldots,x_n:A_n
\]
of pairs $x_i:A_i$ consisting of a variable~$x_i$ and a type $A_i$, what we
sometimes write $\Gamma=(x_i:A_i)_{1\leq i\leq n}$, the empty context being
denoted $\cempty$.

\begin{definition}
  The \emph{dimension} $\dim(A)$ of a type $A$ is the natural number defined
  inductively by
  \[
  \dim(\Obj)=0
  \qquad\text{and}\qquad
  \dim(\Hom Atu)=\dim(A)+1
  \]
  Given a context $\Gamma=(x_i:A_i)_{1\leq i\leq n}$, its dimension is
  $\max\setof{\dim(A_i)}{1\leq i\leq n}$.
\end{definition}

The reader will observe that these definitions are standard for the
construction of a dependent type theory.  We would like, however, to
emphasize the geometric intuition that this syntax naturally captures,
specifically, that of a finite globular set: a variable~$x$ corresponds
to a cell and its type~$A$ indicate its dimension (namely, $\dim(A)$)
as well as its source and its target. For instance, a variable
$x:\Obj$ corresponds to a $0$-cell and a variable $x:\Hom Atu$
corresponds to a $(\dim(A)+1)$-cell whose source is~$t$ and target
is~$u$.

\begin{definition}
  The set of \emph{free variables} is defined
  \begin{itemize}
  \item on terms by
    \[
    \FV(x)=\set{x}
    \]
  \item on substitutions by
    \[
    \FV(\sempty)=\emptyset
    \qquad
    \FV(\sext\sigma t)=\FV(\sigma)\cup\FV(t)
    \]
  \item on types by
    \[
    \FV(\Obj)=\emptyset
    \qquad
    \FV(\Hom Atu)=\FV(A)\cup\FV(t)\cup\FV(u)
    \]
  \item on contexts by
    \[
    \FV(\cempty)=\emptyset
    \qquad
    \FV(\Gamma,x:A)=\FV(\Gamma)\cup\set{x}\cup\FV(A)
    \]
  \end{itemize}
\end{definition}


\subsection{Typing rules}
As usual in dependent type theories, we consider four different kinds of
judgments whose informal interpretation is the following:
\begin{itemize}
\item $\Gamma\vdash$ means $\Gamma$ is a context,
\item $\Gamma\vdash A$ means $A$ is a type in context $\Gamma$,
\item $\Gamma\vdash t:A$ means $t$ has type $A$ in context $\Gamma$,
\item $\Gamma\vdash\sigma:\Delta$ means that $\sigma$ is a substitution of type
  $\Delta$ in context $\Gamma$.
\end{itemize}
A judgment holds when it is derivable using the following inference rules, which
we call the \emph{globular type theory}.

\subsubsection{Rules for types}
\[
\inferrule{\Gamma\vdash}{\Gamma\vdash\Obj}
\qquad\qquad
\inferrule{\Gamma\vdash t:A\\\Gamma\vdash u:A}{\Gamma\vdash\Hom Atu}
\]

\subsubsection{Rules for terms}
\[
\inferrule{\Gamma,x:A\vdash}{\Gamma,x:A\vdash x:A}
\qquad\qquad
\inferrule{\Gamma\vdash t:B}{\Gamma,x:A\vdash t:B}
\]
where we suppose $x\not\in\FV(t)\cup\FV(B)$ in the second rule

\subsubsection{Rules for contexts}
\[
\inferrule{\null}{\cempty\vdash}
\qquad\qquad
\inferrule{\Gamma\vdash A}{\Gamma,x:A\vdash}
\]
where we suppose $x\not\in\FV(\Gamma)$ in the second rule

\subsubsection{Rules for substitutions}
\[
\inferrule{\null}{\Gamma\vdash\sempty:\cempty}
\qquad\qquad
\inferrule{\Delta\vdash\sigma:\Gamma\\\Gamma\vdash A\\\Delta\vdash t:\csubst A\sigma\Gamma}{\Delta\vdash\sext\sigma t:(\Gamma,x:A)}
\]
The notation for the application of substitutions $\csubst A\sigma\Gamma$
is explained in next section.




\begin{lemma}
  The following can be shown.
  \begin{itemize}
  \item If $\Gamma\vdash t:A$ holds then $\Gamma\vdash A$ holds.
  \item If $\Delta\vdash\sigma:\Gamma$ holds then $\Gamma\vdash$ holds.
  \item If $\Gamma\vdash A$ holds then $\FV(A)\subseteq\FV(\Gamma)$.
  \item If $\Gamma\vdash$ holds then $\FV(\Gamma)=\set{x_1,\ldots,x_n}$, with
    $\Gamma=(x_i:A_i)_{1\leq i\leq n}$.
  \end{itemize}
\end{lemma}

\noindent
Finally, the following lemma allows us to identify derivable judgments
and their derivations.

\begin{lemma}
  A judgment can be derived in at most one way.
\end{lemma}

\subsection{Substitutions}
\label{sec:substitution}
Consider a context $\Gamma=x_1{:}A_1,\ldots,x_n{:}A_n$ and a substitution
$\sigma=\slist{t_1,\ldots,t_m}$ such that $\Delta\vdash\sigma:\Gamma$ holds. In
this case, we necessarily have $m=n$. Given a type~$A$ we write
$\csubst A\sigma\Gamma$ for the type obtained from $A$ by replacing each
variable $x_i$ by the term~$t_i$; given a term $t$, the term
$\csubst t\sigma\Gamma$ is defined similarly. More formally, we have
%
\[
\csubst\Obj\sigma\Gamma=\Obj
\qquad\qquad
\csubst{(\Hom Atu)}\sigma\Gamma=\Hom{\csubst A\sigma\Gamma}{\csubst t\sigma\Gamma}{\csubst u\sigma\Gamma}
\]
on types, and
\[
\csubst{x_i}\sigma\Gamma=t_i
\]
on terms. Application of substitutions is compatible with typing:

\begin{lemma}
  \label{lem:subst-typ}
  The following rule is admissible:
  \[
  \inferrule{\Delta\vdash\sigma:\Gamma\\\Gamma\vdash t:A}{\Delta\vdash\csubst t\sigma\Gamma:\csubst A\sigma\Gamma}
  \]
\end{lemma}

Given another substitution~$\tau$ such that $\Upsilon\vdash\tau:\Delta$ holds,
we write $\sigma\circ\tau$ for the \emph{composite substitution}
\[
\sigma\circ\tau\qeq \slist{t_1[\tau],\ldots,t_n[\tau]}
\]
and given a context $\Gamma=(x_i:A_i)_{1\leq i\leq n}$, the associated
\emph{identity substitution} is
\[
\id_\Gamma\qeq\slist{x_1,\ldots,x_n}
\]

\begin{lemma}
  The following rules are admissible
  \[
  \inferrule{\Upsilon\vdash\tau:\Delta\\\Delta\vdash\sigma:\Gamma}{\Upsilon\vdash\sigma\circ\tau:\Gamma}
  \qquad\qquad
  \inferrule{\Gamma\vdash}{\Gamma\vdash\id_\Gamma:\Gamma}
  \]
  Moreover, composition is associative and admits identities as neutral
  elements.
\end{lemma}





\subsection{The syntactic category}
We are now in position to define the category generated by this type theory.

\begin{definition}
  The \emph{syntactic category}~$\Sglob$ associated to this theory is the
  category whose
  \begin{itemize}
  \item objects are contexts~$\Gamma$ such that $\Gamma\vdash$ holds,
  \item morphisms $\sigma:\Delta\to\Gamma$ are substitutions such that
    $\Delta\vdash\sigma:\Gamma$ holds.
  \end{itemize}
\end{definition}

\noindent
The following proposition shows that, in fact, contexts can be considered as a
notation for finite globular sets.

\begin{proposition}
  \label{prop:Sglob-fin-gset}
  The category $\Sglob$ is equivalent to the category~$\FinGSet^\op$.
\end{proposition}
\begin{proof}
  We construct a functor $F:\Sglob\to\FinGSet^\op$ as follows. Given a context
  $\Gamma=(x_i:A_i)_{1\leq i\leq k}$, we define $F\Gamma$ to be the globular
  set~$G^\Gamma$ with $G^\Gamma_n=\setof{(x_i,i)}{\dim(A_i)=n}$ (the second
  component ensures that two different instances of a variable in~$\Gamma$
  gives rise to two distinct cells). Given an $(n+1)$-cell $(x_i,i)$, the type of
  $x_i$ is of the form $A_i=\Hom Ayz$ and we define its source and target as
  $s_n(x_i,i)=y$ and $t_n(x_i,i)=z$. The fact that the globular identities
  \eqref{eq:gset} hold can be shown by induction on the derivation of
  $\Gamma\vdash$.


  Suppose given a morphism $\sigma:\Delta\to\Gamma$ with
  $\Gamma=(x_i:A_i)_{1\leq i\leq m}$ and $\Delta=(y_i:B_i)_{1\leq i\leq n}$. The
  substitution is of the form $\sigma=\slist{z_1,\ldots,z_m}$ and for every
  index $i$ with $1\leq i\leq m$ there is an index $j_i$ with $1\leq j_i\leq n$
  such that $z_i=y_{j_i}$. We then define
  $f^\sigma_{\dim(A_i)}(x_i,i)=(y_{j_i},j_i)$. In order to formally account for
  the case where a same variable occurs multiple times in~$\Delta$, this
  definition should in fact be performed by induction on the derivation of
  $\Delta\vdash\sigma:\Gamma$, in the expected way. By a similar induction, the
  morphism $f^\sigma$ can be shown to be a morphism of globular sets.

  The functor $F$ is faithful since a substitution $\sigma$ can be recovered
  from~$f^\sigma$: we have $\sigma=\slist{f(x_1,1),\ldots,f(x_n,n)}$. The
  functor is also full since for any morphism $f$, the substitution~$\sigma$
  defined as previously can be shown to be such that $\Delta\vdash\sigma:\Gamma$
  holds, by induction on~$\Gamma\vdash$. Finally, the functor~$F$ is essentially
  surjective: given a globular set~$G$ and an enumeration
  $G_\infty=\setof{x_i}{1\leq i\leq m}$ of all cells compatible with dimensions
  (\ie a total ordering of cells such that $i\leq j$ implies
  $\dim(x_i)\leq\dim(x_j)$), $G$ is isomorphic to the image of the context
  $(x_i:A_i)_{1\leq i\leq m}$ where $A_i=\Obj$ if $x_i$ is a $0$-cell, and
  $A_k=\Hom{A_i}{x_i}{x_j}$ if $x_k$ is an $(n+1)$-cell with $s_n(x_k)=x_i$ and
  $t_n(x_k)=x_j$.
\end{proof}


\begin{example}
  \label{ex:ctx-gset}
  The context corresponding to the globular set of \exr{gset} is
  \[
  x:\Obj,y:\Obj,z:\Obj,f:\Hom\Obj xy,g:\Hom\Obj xy,h:\Hom\Obj zy,\alpha:\Hom{\Hom\Obj xy}fg
  \]
  Other contexts also correspond to this globular set (for instance the one
  obtained by permuting $x$ and $y$), but they are isomorphic to this one.
  The substitution corresponding to the only morphism
  \[
  \xymatrix{
    x\ar[r]^f&y
  }
  \qquad\to\qquad
  \xymatrix{
    z\ar@(ur,dr)^g
  }
  \]
  is
  $
  z:\Obj,g:\Hom\Obj zz\vdash\slist{z,z,g}:(x:\Obj,y:\Obj,f:\Hom\Obj xy)
  $.
\end{example}

\subsection{Models}
\label{sec:models}
Let us briefly recall the notion of model for a dependent type theory in
categories with families~\cite{dybjer1995internal}. The category~$\Fam$ has
families of sets~$(E_x)_{x\in B}$ as objects (the set $B$ is called the \emph{base} and
the set $E_x$ the \emph{fiber} over~$x$) and a morphism
$(E_x)_{x\in B}\to (E'_x)_{x\in B'}$ consists in functions $f:B\to B'$ and
$(g_x:E_x\to E'_{f(x)})_{x\in B}$.

\newcommand{\Ty}{\operatorname{Ty}}
\newcommand{\Tm}{\operatorname{Tm}}

A \emph{category with families} (or \emph{cwf}) consists of a category~$\C$
together with a functor
\[
T\qcolon\C^\op\qto\Fam
\]
Given an object $\Gamma$ of $\C$, we write
\[
TA\qeq(\Tm^\Gamma_A)_{A\in\Ty^\Gamma}
\]
\ie $\Ty^\Gamma$ for the base of $T\Gamma$ and $\Tm^\Gamma_A$ for the
fibers. Similarly, given a morphism $\sigma:\Delta\to\Gamma$ in~$\C$, we write
$\Ty^\sigma:\Ty^\Gamma\to\Ty^\Delta$ and
$\Tm^\sigma_A:\Tm^\Gamma_A\to\Tm^\Delta_{\Ty^\sigma(A)}$ for the functions
constituting its image. The category~$\C$ should moreover satisfy the following
axioms: it should have a terminal object~$\emptyset$ and a \emph{context
  comprehension} operation which to an object $\Gamma$ of~$\C$ and an element
$A\in\Ty^\Gamma$ associates an object $(\Gamma,A)$, a morphism
$\pi:(\Gamma,A)\to\Gamma$ and an element $p\in\Tm^\Delta_{\Ty^\pi(A)}$, in a way
such that for every morphism object~$\Delta$, morphism $\sigma:\Delta\to\Gamma$
and element $t\in\Tm^\Delta_{\Ty^\sigma(A)}$ there is a unique morphism
$\slist{\sigma,t}:\Delta\to(\Gamma,A)$ such that
$\pi\circ\slist{\sigma,t}=\sigma$ and
$\Tm^{\slist{\sigma,t}}_{\Ty^\pi(A)}(p)=t$.
A \emph{morphism} between cwfs $T:\C^\op\to\Fam$ and $T':\C'^\op\to\Fam$
consists of a functor $F:\C^\op\to\C'^\op$ and a natural transformation
$\phi:T\to T'\circ F$, preserving the terminal object and context
comprehension on the nose. Given two morphisms $(F_1,\phi_1):T\to T'$ and
$(F_2,\phi_2):T\to T$, a \emph{2\nbd-morphism} $\theta:(F_1,\phi_1)\To(F_2,\phi_2)$
is a natural transformation $\theta:F_1\to F_2$ such that
$T\theta\circ\phi_1=\phi_2$.

Typically, the syntactic category $\Sglob$ is canonically a cwf when equipped
with the functor $T:\Sglob^\op\to\Fam$ such that for a context $\Gamma$, we have
$\Ty^\Gamma$ the set of types $A$ such that $\Gamma\vdash A$ and $\Tm^\Gamma_A$
the set of terms~$t$ such that $\Gamma\vdash t:A$ (thus the notations
above). The category $\Set$ is also canonically a cwf with the functor which to
a set~$X$ associates the family with $\Ty^X$ being the collection of functions
$f:Y\to X$ with~$X$ as codomain and $\Tm^X_f$ being the set of sections of~$f$.

A \emph{model} of the globular type theory is a morphism of cwfs $\Sglob\to\C$
for some cwf~$\C$. A \emph{set-theoretic model} is a model where~$\C=\Set$.


%

%



\begin{proposition}
  The category of set-theoretic models of $\Sglob$ is equivalent to the
  category~$\GSet^\op$.
\end{proposition}


\section{Pasting schemes}
\label{sec:pasting-schemes}
The main contribution of this article is to provide a simple description of
pasting schemes, encoding a collection of composable cells.
For instance, in a $1$\nbd-category, we expect the diagram
\[
\vxym{
  x_0\ar[r]^{f_1}&x_1\ar[r]^{f_2}&x_2\ar[r]^{f_3}&x_3\ar[r]^{f_4}&x_4
}
\]
to give rise to a unique composite (it does not depend on the order in which the
morphisms are pairwise composed), but diagrams such as
\[
\xymatrix{
  x\ar@/^/[r]^f&\ar@/^/[l]^gy&z
}
\qquad\quad\text{or}\qquad\quad
\xymatrix{
  x\ar[r]^f&y&z\ar[l]_g
}
\]
are not expected to be composed. A formal description of these pasting schemes
in higher dimensions is not easy. It was achieved, in the globular setting, by
Grothendieck~\cite{grothendieck1983pursuing} using abstract categorical techniques
and studied combinatorially by Batanin~\cite{batanin1998monoidal}.

\begin{example}
  \label{ex:ps}
  The following diagram is a pasting scheme in a $2$- (or higher-) category:
  \[
  \vxym{
    x
    \ar@/^3ex/[r]^f
    \ar@{{}{ }{}}@/^1.5ex/[r]|{\phantom\alpha\Downarrow\alpha}
    \ar[r]|{f'}
    \ar@{{}{ }{}}@/_1.5ex/[r]|{\phantom\beta\Downarrow\beta}
    \ar@/_3ex/[r]_{f''}
    &
    y\ar[r]^g&z\ar[r]^h&w
  }
  \]
\end{example}

\noindent
The pasting scheme above corresponds to a globular set, which can be obtained as
the following colimit of disks, where the dotted arrows correspond to the
obvious monomorphisms of globular sets:
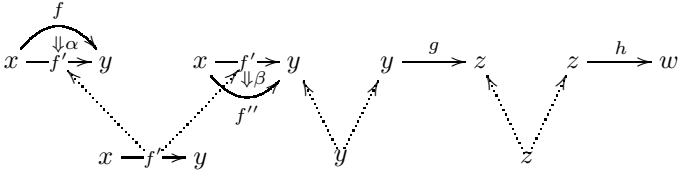
\[
\vxym{
  x
  \ar@/^3ex/[r]^f
  \ar@{{}{ }{}}@/^1.5ex/[r]|{\phantom\alpha\Downarrow\alpha}
  \ar[r]|{f'}="f'0"
  &
  y
  &
  x
  \ar[r]|{f'}="f'1"
  \ar@{{}{ }{}}@/_1.5ex/[r]|{\phantom\beta\Downarrow\beta}
  \ar@/_3ex/[r]_{f''}
  &
  \ar@{}|{}="y0"y
  &
  \ar@{}|{}="y1"y\ar[r]^g&z\ar@{}|{}="z0"
  &
  \ar@{}|{}="z1"z\ar[r]^h&w
  \\
  &x\ar[r]|{f'}="f'"&y&\ar@{}[r]|{\displaystyle y}="y"&&\ar@{}[r]|{\displaystyle z}="z"&
  \ar@{.>}"f'";"f'0"
  \ar@{.>}"f'";"f'1"
  \ar@{.>}"y";"y0"
  \ar@{.>}"y";"y1"
  \ar@{.>}"z";"z0"
  \ar@{.>}"z";"z1"
}
\]
that is, to the colimit of the diagram
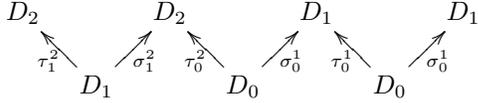
\[
\xymatrix@C=2ex@R=3ex{
  D_2&&D_2&&D_1&&D_1\\
  &
  \ar[ul]^*-{\scriptstyle\tau^2_1}
  D_1
  \ar[ur]_*-{\scriptstyle\sigma^2_1}
  &&
  \ar[ul]^*-{\scriptstyle\tau^2_0}
  D_0
  \ar[ur]_*-{\scriptstyle\sigma^1_0}
  &&
  \ar[ul]^*-{\scriptstyle\tau^1_0}
  D_0
  \ar[ur]_*-{\scriptstyle\sigma^1_0}
}
\]
of globular sets.
The idea of Grothendieck's definition~\cite{grothendieck1983pursuing} is that
pasting schemes are precisely diagrams which can be obtained by such colimits,
which are called globular sums. Our presentation given below is largely inspired
of the work of Maltsiniotis~\cite{maltsiniotis2010grothendieck} and
Ara~\cite{ara:these}.

\subsection{Globular extensions}
\label{sec:gext}
A \emph{globular category} consists of a category~$\C$ together with a
functor~$\Glob\to\C$, \ie it is a category equipped a notion of ``disk'': we
write $D_n$ for the image of~$n$ and denote in the same way the morphisms
in~$\Glob$ and their image. A morphism~$f$ of~$\C$ is \emph{globular} when it is
the image of one in~$\Glob$. A \emph{globular sum} is the colimit of a diagram
of the form
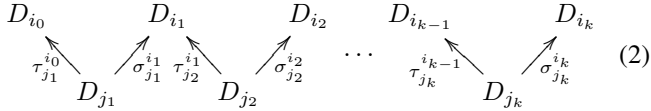
\begin{equation}
  \label{eq:gsum}
  \vcenter{
    \xymatrix@C=1.1ex@R=3ex{
      D_{i_0}&&D_{i_1}&&D_{i_2}&\ar@{}[d]|{\displaystyle\ldots}&D_{i_{k-1}}&&D_{i_k}\\
      &
      \ar[ul]^*-{\scriptstyle\tau^{i_0}_{j_1}}
      D_{j_1}
      \ar[ur]_*-{\scriptstyle\sigma^{i_1}_{j_1}}
      &&
      \ar[ul]^*-{\scriptstyle\tau^{i_1}_{j_2}}
      D_{j_2}
      \ar[ur]_*-{\scriptstyle\sigma^{i_2}_{j_2}}
      &&&&
      \ar[ul]^*-{\scriptstyle\tau^{i_{k-1}}_{j_k}}
      D_{j_k}
      \ar[ur]_*-{\scriptstyle\sigma^{i_k}_{j_k}}
    }
  }
\end{equation}
in~$\C$, with $k\geq 0$ (the diagram cannot be empty). A globular
category~$\Glob\to\C$ is a \emph{globular extension} when all globular sums
exist.
The category~$\GExt$ of globular extensions is the subcategory of the slice
category $\Glob/\Cat$ whose objects are globular extensions and morphisms are
functors preserving globular sums.

By definition, there is a forgetful functor $\GExt\to\Glob/\Cat$, sending a
globular extension to the underlying globular category, which admits a left
(2-)adjoint. In particular, there is a free globular extension on the globular
category given by the identity functor $\Glob\to\Glob$: this category is
called~$\T_0$ (and we have a functor $\Glob\to\T_0$). Alternatively, it can be
characterized as follows:

\begin{definition}
  The globular extension $\Glob\to\T_0$ is the one such that for every globular
  extension~$\Glob\to\C$ there exists a morphism of globular extensions
  $\T_0\to\C$, which is unique up to isomorphism.
\end{definition}

\noindent
Intuitively, the category $\T_0$ is the category obtained by considering formal
disks and freely completing it under globular sums. Its objects thus correspond
to pasting schemes, but we restrict this terminology to the alternative
description of those objects as globular sets, given in next section.

\subsection{Pasting schemes}
We now recall the more usual description of pasting schemes. As for any presheaf
category (see~\cite{maclane2012sheaves}) the category of globular sets
$\GSet=\hat\Glob$ is the free cocompletion of the category~$\Glob$, meaning that
it is cocomplete and that for any functor $F:\Glob\to\C$, where~$\C$ is a
cocomplete category, there exists a unique (up to isomorphism) functor
$\hat F:\hat\Glob\to\C$ which preserves colimits and makes the diagram
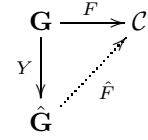
\[
\vxym{
  \Glob\ar[d]_-Y\ar[r]^F&\C\\
  \hat\Glob\ar@{.>}[ur]_{\hat F}
}
\]
commute, where~$Y:\Glob\to\hat\Glob$ is the Yoneda embedding. Note that this
functor makes $\hat\Glob$ into a globular category.
Since the category~$\T_0$ is only required to have globular sums (and not all
colimits), we can expect to recover~$\T_0$ as a full subcategory of~$\hat\Glob$
consisting of globular sums of representables:

\begin{proposition}[{\cite[Proposition~2.2.1]{ara:these}}]
  The category~$\T_0$ is the full subcategory of $\GSet$ whose objects are
  globular sums of representables.
\end{proposition}

\noindent
The globular sets which are the objects of~$\T_0$, as described by the above
proposition, are called \emph{pasting schemes}.

\begin{example}
  \label{ex:ps-gset}
  The globular set~$G$ corresponding to \exr{ps} is the globular set with
  \[
  G_0=\set{x,y,z,w}
  \qquad
  G_1=\set{f,f',f'',g,h}
  \qquad
  G_2=\set{\alpha,\beta}
  \]
  and $G_n=\emptyset$ for $n\geq 3$, with source and targets as indicated on the
  figure.
\end{example}

Finally, we recall how the source and target of a pasting scheme can be
described, see~\cite{maltsiniotis2010grothendieck} for details.

\begin{definition}
  \label{def:ps-boundary}
  Suppose given a pasting scheme~$G$, which can be obtained as a globular sum of
  the form~\eqref{eq:gsum}.
  Given an integer $n\in\N$, its \emph{$n$-boundary} $\partial_nG$ is the
  colimit of the diagram obtained from~\eqref{eq:gsum} by replacing each object
  $D_{i_m}$ (\resp $D_{j_m}$) by $D_{\min(i_m,n)}$ (\resp
  $D_{\min(j_m,n)}$). Moreover, there are two canonical morphisms
  $\sigma_i^G,\tau_i^G:\partial_iG\to G$ exhibiting $\partial_iG$ as the source
  and target of~$G$ respectively.
\end{definition}

\begin{example}
  Consider the following pasting scheme $G$:
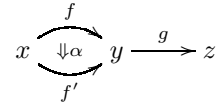
  \[
  \vxym{
    x\ar@/^2ex/[r]^f\ar@/_2ex/[r]_{f'}\ar@{}[r]|{\Downarrow\alpha}&y\ar[r]^g&z
  }
  \]
  Its $1$-source and $1$-target are respectively
  
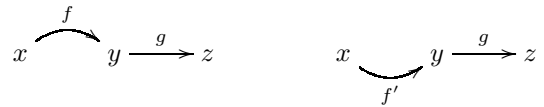
\[
  \xymatrix{
    x\ar@/^2ex/[r]^f&y\ar[r]^g&z
  }
  \qquad\qquad
  \xymatrix{
    x\ar@/_2ex/[r]_{f'}&y\ar[r]^g&z
  }
  \]
  (both are $\partial_1G$, and are in particular isomorphic, but the different
  namings make clear the respective inclusions $\sigma^G_1$ and~$\tau^G_1$) and
  its $0$-source and $0$-target are respectively $x$ and~$z$.
\end{example}

\subsection{A characterization}
\label{sec:tl}
We now introduce a characterization of pasting schemes, which is apparently new
and turns out to be very convenient to work with in the following. First, note
that since pasting schemes are finite colimits of disks, which are finite
globular sets, and colimits are computed pointwise, we have

\begin{lemma}
  Pasting schemes are finite globular sets, \ie $\T_0$ is a full subcategory of
  $\FinGSet$.
\end{lemma}

\noindent
We now introduce a relation $\triangleleft$ which expresses when a cell is
``before'' another in a globular set. Similar relations have already been
considered before, \eg for pasting schemes~\cite{street1991parity}.

\begin{definition}
  Given a globular set~$G$, we define the relation~$\triangleleft$ on its set
  $G_\infty$ of cells as the transitive closure of the relation such that for
  every $(n+1)$-cell~$x\in G_\infty$ one has
  \[
  s_n(x)
  \quad\triangleleft\quad
  x
  \quad\triangleleft\quad
  t_n(x)
  \]
\end{definition}

\begin{example}
  \label{ex:ps-tl}
  In the globular set of \exr{ps} (see \exr{ps-gset}), the relation is
  \[
  x\triangleleft f\triangleleft\alpha\triangleleft f'\triangleleft\beta\triangleleft f''\triangleleft y\triangleleft g\triangleleft z\triangleleft h\triangleleft w
  \]
\end{example}

\begin{example}
  \label{ex:gset-tl}
  In the globular set of \exr{gset} (which is not a pasting scheme), the
  relation is
  \[
    x\triangleleft f\triangleleft\alpha\triangleleft g\triangleleft y\triangleright z\triangleright h
  \]
  (a partial order with~$y$ as maximal element).
\end{example}

\begin{example}
  \label{ex:disk-tl}
  The relation on the globular set~$D_n$ is
  \begin{equation}
    \label{eq:disk-tl}
    x_0^-\triangleleft x_1^-\triangleleft\ldots\triangleleft x_{n-1}^-
    \triangleleft x_n\triangleleft
    x_{n-1}^+\triangleleft\ldots\triangleleft x_1^+\triangleleft x_0^+
  \end{equation}
\end{example}

\noindent
Clearly, the relation is preserved by morphisms:

\begin{lemma}
  \label{lem:tl-pres}
  For every morphism $f:G\to G'$ of globular sets and cells $x,y\in G_\infty$
  such that $x\triangleleft y$, we have $f(x)\triangleleft f(y)$.
\end{lemma}

\begin{theorem}
  \label{thm:ps-tl}
  The pasting schemes are the non-empty finite globular sets~$G$ which are
  \emph{$\triangleleft$-linear}, meaning that for every cells $x,y\in G_\infty$,
  \[
  x\triangleleft y\quad\text{or}\quad y\triangleleft x
  \quad\qquad\text{iff}\qquad\quad
  x\neq y
  \]
  This condition is equivalent to the reflexive closure of $\triangleleft$ being
  a total order on the cells of~$G$.
\end{theorem}
\begin{proof}
  We first show that the pasting schemes satisfy the linearity condition, by
  recurrence on~$k\in\N$, the number of peaks in a diagram~\eqref{eq:gsum} whose
  colimit is the pasting scheme. We also show inductively that, if we call $G$
  the colimit and $\iota_k:D_{i_k}\to G$ the canonical arrow we have that the
  successors \wrt $\triangleleft$ of $\iota_k(x_{j_k})$ in~$G$ are precisely
  \begin{equation}
    \label{eq:tl-lin-right}
    \iota_k(x_{j_k})\triangleleft\iota_k(x_{j_k-1}^+)\triangleleft\ldots\triangleleft\iota_k(x_1^+)\triangleleft\iota_k(x_0^+)
  \end{equation}
  where $x_{j_k}$ denotes the top-dimensional cell of $D_{j_k}$: intuitively,
  this cell has its successors unchanged after applying~$\iota_k$, or
  equivalently its only successors are its iterated target faces.
  When the pasting scheme is a disk $D_{i_0}$, the $\triangleleft$-linearity
  condition is satisfied, see \exr{disk-tl}, and the canonical colimiting arrow
  $\iota_0:D_{i_0}\to D_{i_0}$ is the identity and thus satisfies
  \eqref{eq:tl-lin-right}. Otherwise, suppose given a
  diagram~\eqref{eq:gsum}. Since this diagram is finite, connected and simply
  connected, its colimit can be computed using iterated
  pushouts~\cite{pare1990simply}. We thus obtain a diagram of the form
  \[
  \vcenter
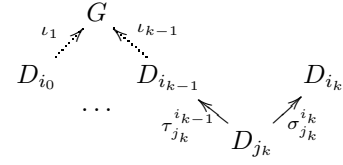
{
    \xymatrix@C=1.1ex@R=2ex{
      &
      G
      \\
      D_{i_0}
      \ar@{.>}[ur]^*{\scriptstyle\iota_1}
      &\ar@{}[d]|{\displaystyle\ldots}&
      \ar@{.>}[ul]_*{\scriptstyle\iota_{k-1}}
      D_{i_{k-1}}&&D_{i_k}\\
      &&&
      \ar[ul]^*-{\scriptstyle\tau^{i_{k-1}}_{j_k}}
      D_{j_k}
      \ar[ur]_*-{\scriptstyle\sigma^{i_k}_{j_k}}
    }
  }
  \]
  We thus consider the colimit~$G$ of the subdiagram obtained by excluding
  $D_{j_k}$ and $D_{i_k}$, with colimiting cocone formed by the morphisms
  $\iota_i$. By induction hypothesis, the set $G$ is $\triangleleft$-linear and
  $\iota_{k-1}$ satisfies~\eqref{eq:tl-lin-right}. The globular sum we are
  interested in is the pushout of $\iota_{k-1}\circ\tau^{i_{k-1}}_{j_k}$ and
  $\sigma^{i_k}_{j_k}$. Because of the order~\eqref{eq:disk-tl} of $D_{i_k}$,
  one easily shows that the effect of the pushout is to ``insert'' the cells
  \[
  x_{j_k+1}^-\triangleleft x_{j_k+2}^-\triangleleft\ldots\triangleleft x_{i_k-1}^-
  \triangleleft x_{i_k}\triangleleft
  x_{i_k-1}^+\triangleleft x_{i_k-2}^+\triangleleft\ldots\triangleleft x_{j_k}^+
  \]
  of $D_{i_k}$ in~\eqref{eq:tl-lin-right} between $\iota_{k-1}(x_{j_k}^+)$ and
  $\iota_{k-1}(x_{j_k-1}^+)$, from which one concludes that the resulting
  globular set is $\triangleleft$-linear and \eqref{eq:tl-lin-right} is
  satisfied.

  Conversely, suppose given a $\triangleleft$-linear finite globular
  set~$G\in\hat\Glob$. We write $x_1,\ldots,x_k$ for the top-dimensional
  elements of~$G$ in the order given by $\triangleleft$, \ie
  $x_1\triangleleft\ldots\triangleleft x_k$.
  As any presheaf, $G$ can be obtained as the colimit of
  representables~\cite{maclane2012sheaves}:
  \[
  G\qeq\colim_{(n,x)\in\El G} D_n
  \]
  where $\El G$ denotes the category of elements of~$G$. A careful examination
  of this diagram shows, by recurrence on $k$, that this colimit is the same as
  the one of the globular sum of the form~\eqref{eq:gsum}, where $i_n$ is then
  dimension of the top-dimensional cell~$x_n$ and $j_{n+1}$ is the greatest
  integer such that $t^{i_n}_{j_{n+1}}(x_n)=s^{i_{n+1}}_{j_{n+1}}$ (such an
  integer necessarily exists because $x_n$ and $x_{n+1}$ are consecutive
  top-dimensional elements \wrt $\triangleleft$).
\end{proof}

\begin{example}
  From \exr{ps-tl} (\resp \ref{ex:gset-tl}), one sees that the globular set of
  \exr{ps} (\resp \ref{ex:gset}) is a pasting scheme (\resp not a pasting
  scheme).
\end{example}

\begin{example}
  To illustrate the first part of the proof, in the case of the globular sum of
  \exr{ps}, the linear orders obtained by iteratively computing the colimit
  using pushouts are
  \[
  \begin{array}{l}
    x\triangleleft f\triangleleft\alpha\triangleleft f'\triangleleft y\\
    x\triangleleft f\triangleleft\alpha\triangleleft f'\triangleleft\beta\triangleleft f''\triangleleft y\\
    x\triangleleft f\triangleleft\alpha\triangleleft f'\triangleleft\beta\triangleleft f''\triangleleft y\triangleleft g\triangleleft z\\
    x\triangleleft f\triangleleft\alpha\triangleleft f'\triangleleft\beta\triangleleft f''\triangleleft y\triangleleft g\triangleleft z\triangleleft h\triangleleft w\\
  \end{array}
  \]
\end{example}

\noindent
This result enables one to immediately draw some interesting consequences:

\begin{lemma}
  A morphism between pasting schemes is necessarily a monomorphism and the only
  automorphism of a pasting scheme is the identity.
\end{lemma}
\begin{proof}
  Suppose given a morphism $f:G\to G'$ between pasting schemes and consider two
  cells $x,y\in G_n$ such that $f_n(x)=f_n(y)$. If $x\neq y$ then, by
  \thmr{ps-tl}, $x\triangleleft y$ or $y\triangleleft x$ and by \lemr{tl-pres}
  on has $f_n(x)\triangleleft f_n(y)=f_n(x)$, which is excluded by
  \thmr{ps-tl}. This argument is the reason why we did not define
  $\triangleleft$ as a preorder, \ie close it under reflexivity. The other
  property uses similar arguments.
\end{proof}



\subsection{Batanin trees}
In order to make a connection with other works on the subject, we
briefly recall here another representation of pasting schemes as trees
introduced by Batanin~\cite{batanin1998monoidal}. The correspondence
with pasting schemes is detailed
in~\cite{berger2002cellular,ara:these}, we recall it here in order to
explain why the linear ordering of cells was to be expected.
In this section, we consider finite planar rooted trees:

\begin{definition}
  A \emph{tree}~$T$ consists of a family of sets $(T_n)_{n\in\N}$, whose
  elements are called $n$-\emph{vertices}, together with a family of functions
  $d_n:T_{n+1}\to T_n$ and a total order on the set~$p_n^{-1}(x)$ of
  \emph{children} of~$x$, for any vertex $x\in T_n$, such that $T_0$ is a
  singleton and $\coprod_{n\in\N}T_n$ is finite.
\end{definition}

\noindent
%
Given a vertex~$x\in T_n$, we write
\[
q_n(x)
\qeq
\set{-}\uplus p_n^{-1}(x)\uplus\set{+}
\]
for the totally ordered set $p_n^{-1}(x)$ extended with a new minimal
element~$-$ and a new maximal element~$+$. A \emph{sector} of~$x$ consists in
two consecutive elements of $q_n(x)$. To any tree $T$, one can associate a
finite globular set~$T_*$ whose $n$-cells are the sectors associated to its
$n$-vertices. The source of a sector $(y_1,y_2)$ of~$x\in T_{n+1}$ is the sector
$(x',x)$ of $f_n(x)$, where $x'$ is the preceding element of~$x$ in
$q_n(f_n(x))$; targets are defined similarly.

By suitably defining the morphisms between trees (which is slightly more involved than
one might expect), this operation extends to a functor from the
category of trees to the category of pasting schemes, which can be shown to be
an equivalence, \ie pasting schemes can be represented as trees:

\begin{proposition}[\cite{berger2002cellular,ara:these}]
  The category~$\T_0$ is equivalent to the category of trees and suitable
  morphisms.
\end{proposition}

\begin{example}
  Consider the tree~$T$ with $T_0=\set{x^0_0}$, $T_1=\set{x^1_0,x^1_1,x^1_2}$,
  $T_2=\set{x^2_0,x^2_1}$, $p_1(x^2_i)=x^1_0$ and $p_0(x^1_i)=x^0_0$:
  
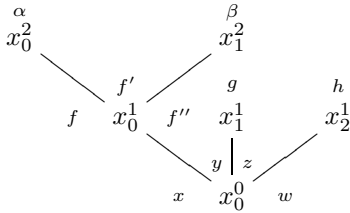
\vspace{-6ex}
  \[
  \xymatrix@R=3ex{
    &&\\
    \ar@{}[u]|(.4)\alpha
    x^2_0
    &&
    \ar@{}[u]|(.4)\beta
    x^2_1
    &\\
    &
    \ar@{}[l]|f
    \ar@{}[u]|(.4){f'}
    \ar@{}[r]|{f''}
    x_0^1
    \ar@{-}[ul]
    \ar@{-}[ur]
    &
    x_1^1\ar@{}[u]|(.4)g
    &
    x_2^1\ar@{}[u]|(.4)h
    \\
    &&
    \ar@{}[l]|x
    \ar@{{}{ }{}}@/^/[u]|(.4)y
    \ar@{{}{ }{}}@/_/[u]|(.42)z
    \ar@{}[r]|w
    x_0^0
    \ar@{-}[ul]
    \ar@{-}[u]
    \ar@{-}[ur]
    &
  }
  \]
  We have figured the sectors in small letters. For instance, the sectors
  associated to $x^1_0$ are $f,f',f''$ and the one associated to~$x^2_0$ is
  $\alpha$. The globular set~$T_*$ is precisely the pasting scheme of
  \exr{ps}. The order on cells
  \[
  x\triangleleft f\triangleleft\alpha\triangleleft f'\triangleleft\beta\triangleleft f''\triangleleft y\triangleleft g\triangleleft z\triangleleft h\triangleleft w
  \]
  is precisely the list of sectors encountered if we draw a line around the tree
  starting from the bottom left (think of a child drawing the contour of his
  hand).
\end{example}

\noindent
This ``duality'' between trees and pasting schemes was nicely
explained by Joyal in~\cite{joyal1997disks} where he additionally
introduces a generalization of the category~$\T_0$.

\subsection{A type-theoretic definition of pasting schemes}
\label{sec:ps-tt}
We have seen in \propr{Sglob-fin-gset} that contexts correspond to
finite globular sets. Our aim is now to characterize those contexts
which correspond to pasting schemes in a type-theoretic fashion, that
is to say, using a system of inference rules. It turns out that this
characterization yields canonical forms for pasting schemes: we have
seen in \exr{ctx-gset} that multiple (isomorphic) contexts may
correspond to a same pasting scheme, but our definition singles out
exactly one. Our main tool is the linear order $\triangleleft$ studied
in \secr{tl}. The reader will notice, however, that this order does
not give rise to a well formed context. For instance, consider the
pasting scheme on the left below, whose $\triangleleft$ relation is
shown in the middle:
\[
\xymatrix{
  x\ar[r]^f&y
}
\qquad\qquad
x\triangleleft f\triangleleft y
\qquad\qquad
x:\Obj,f:\Hom\Obj xy,y:\Obj
\]
A direct translation of this order as a ``context'' is shown on the right, but
it is not a well-formed context since the variable~$y$ has to be declared before
$f$ whose type involves~$y$. We will therefore use another enumeration of the
cells and the associated context will in fact be
\[
x:\Obj,y:\Obj,f:\Hom\Obj xy
\]

Note that $y:\Obj,x:\Obj,f:\Hom\Obj xy$ would be another sensible representation
for the pasting scheme, but our typing rules will only accept the first one.

We add two new kinds of judgments to our type theory:
\begin{itemize}
\item $\Gamma\vdashps$ means $\Gamma$ is a context which is a pasting scheme,
\item $\Gamma\vdashps x:A$ means $\Gamma$ is a partial pasting scheme (one which
  is being constructed), with $x$ as ``free output''.
\end{itemize}
The rules for showing that a context is a pasting scheme are the following ones:
\[
\inferrule{\Gamma\vdashps x:\Obj}{\Gamma\vdashps}
\qquad\qquad
\inferrule{\null}{x:\Obj\vdashps x:\Obj}
\]
\[
\inferrule{\Gamma\vdashps x:A}{\Gamma,y:A,f:\Hom Axy\vdashps f:\Hom Axy}
\qquad\qquad
\inferrule{\Gamma\vdashps f:\Hom Axy}{\Gamma\vdashps y:A}
\]
where on the bottom left we suppose $y,f\not\in\FV(\Gamma)$. In the
first line, the first rule allows one to conclude that a partial
pasting scheme is in fact a pasting scheme, whereas the second one
allows one to start constructing a pasting scheme with one
$0$-cell. In the second line, the first rule allows one to attach a new
cell to a pasting scheme, and the second rule to drop the possibility
of attaching a cell to~$f$. A context~$\Gamma$ such that
$\Gamma\vdashps$ holds is called a \emph{ps-context}. Observe that
every ps-context is of odd length.

\begin{example}
  \label{ex:ps-tt}
  The context corresponding to the pasting scheme
  \[
  \vxym{
    x\ar@/^2ex/[r]^f\ar@/_2ex/[r]_{f'}\ar@{}[r]|{\Downarrow\alpha}&y\ar[r]^g&z
  }
  \]
  is derived as follows:
  \[
  \inferrule{
    \inferrule{
      \inferrule{
        \inferrule{
          \inferrule{
            \inferrule{
              \inferrule{
                \inferrule{
                  \null
                }
                {x:\Obj\vdashps x:\Obj}
              }
              {x:\Obj,y:\Obj,f:\Hom\Obj xy\vdashps f:\Hom\Obj xy}
            }
            {\scalebox{.87}{$x:\Obj,y:\Obj,f:\Hom\Obj xy,f':\Hom\Obj xy,\alpha:\Hom{\Hom\Obj xy}f{f'}\vdashps \alpha:\Hom{\Hom\Obj xy}f{f'}$}}
          }
          {\scalebox{.9}{$x:\Obj,y:\Obj,f:\Hom\Obj xy,f':\Hom\Obj xy,\alpha:\Hom{\Hom\Obj xy}f{f'}\vdashps f':\Hom\Obj xy$}}
        }
        {x:\Obj,y:\Obj,f:\Hom\Obj xy,f':\Hom\Obj xy,\alpha:\Hom{\Hom\Obj xy}f{f'}\vdashps y:\Obj}
      }
      {\scalebox{.78}{$x:\Obj,y:\Obj,f:\Hom\Obj xy,f':\Hom\Obj xy,\alpha:\Hom{\Hom\Obj xy}f{f'},z:\Obj,g:\Hom\Obj yz\vdashps g:\Hom\Obj yz$}}
    }
    {\scalebox{.83}{$x:\Obj,y:\Obj,f:\Hom\Obj xy,f':\Hom\Obj xy,\alpha:\Hom{\Hom\Obj xy}f{f'},z:\Obj,g:\Hom\Obj yz\vdashps z:\Obj$}}
  }
  {\scalebox{.9}{$x:\Obj,y:\Obj,f:\Hom\Obj xy,f':\Hom\Obj xy,\alpha:\Hom{\Hom\Obj xy}f{f'},z:\Obj,g:\Hom\Obj yz\vdashps$}}
  \]
  Graphically, it corresponds to constructing the pasting scheme in the
  following way
  
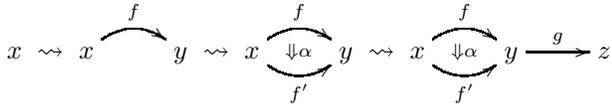
\[
  \xymatrix{
    x
  }
  \rightsquigarrow
  \xymatrix{
    x\ar@/^2ex/[r]^f&y
  }
  \rightsquigarrow
  \xymatrix{
    x\ar@/^2ex/[r]^f\ar@/_2ex/[r]_{f'}\ar@{}[r]|{\Downarrow\alpha}&y
  }
  \rightsquigarrow
  \xymatrix{
    x\ar@/^2ex/[r]^f\ar@/_2ex/[r]_{f'}\ar@{}[r]|{\Downarrow\alpha}&y\ar[r]^g&z
  }
  \]
  Also, note that the variables occurring on the right precisely do so in the
  $\triangleleft$ order when read from top to bottom:
  \[
  x\triangleleft f\triangleleft\alpha\triangleleft f'\triangleleft y\triangleleft g\triangleleft z
  \]
  Finally, in \figr{ps-tt}, we have figured for each sequent
  $\Gamma\vdashps x:A$ the $\triangleleft$ relation (the height of each cell
  corresponding to its dimension), with the cell~$x$ underlined.
  \begin{figure}[ht!]
    \centering
  \begin{tabular}{|c|c|}
    \hline
    $
    \xymatrix@C=1ex@R=1ex{
      {}\phantom{\ul f}\\
      \ul x
    }
    $
    &
    $
    \xymatrix@C=1ex@R=1ex{
      &\ul f\tldr\\
      x\tlur&&y
    }
    $
    \\
    \hline
    $
    \xymatrix@C=1ex@R=1ex{
      &&\ul\alpha\tldr\\
      &f\tlur&&f'\tldr\\
      x\tlur&&&&y
    }
    $
    &
    $
    \xymatrix@C=1ex@R=1ex{
      &&\alpha\tldr\\
      &f\tlur&&\ul{f'}\tldr\\
      x\tlur&&&&y
    }
    $
    \\
    \hline
    $
    \xymatrix@C=1ex@R=1ex{
      &&\alpha\tldr\\
      &f\tlur&&f'\tldr\\
      x\tlur&&&&\ul y
    }
    $
    &
    $
    \xymatrix@C=1ex@R=1ex{
      &&\alpha\tldr\\
      &f\tlur&&f'\tldr&&\ul g\tldr\\
      x\tlur&&&&y\tlur&&z
    }
    $
    \\
    \hline
    $
    \xymatrix@C=1ex@R=1ex{
      &&\alpha\tldr\\
      &f\tlur&&f'\tldr&&g\tldr\\
      x\tlur&&&&y\tlur&&\ul z
    }
    $
    &
    $
    \xymatrix@C=1ex@R=1ex{
      &&\alpha\tldr\\
      &f\tlur&&f'\tldr&&g\tldr\\
      x\tlur&&&&y\tlur&&z
    }
    $
    \\
    \hline
  \end{tabular}    
    \caption{$\triangleleft$ relations of \exr{ps-tt}}
    \label{fig:ps-tt}
  \end{figure}
\end{example}

\begin{lemma}
  Given a context $\Gamma$ such that $\Gamma\vdashps$ holds, $\Gamma\vdash$ also
  holds.
\end{lemma}

\begin{lemma}
  There is at most one way to show a statement of the form $\Gamma\vdashps$ or
  $\Gamma\vdashps x:A$.
\end{lemma}

\noindent
Previous example should make it clear that there is a tight correspondence
between type ps-contexts and pasting schemes as previously defined. First note
that in a ps-context, there is no variable clash:

\begin{lemma}
  In a ps-context $\Gamma=(x_i:A_i)_{1\leq i\leq n}$, for every variables $x_i$
  and $x_j$ such that $x_i=x_j$ we have $i=j$.
\end{lemma}

\noindent
For this reason, when associating a globular set~$G^\Gamma$ to a
ps-context~$\Gamma=(x_i:A_i)_{1\leq i\leq k}$, we can proceed in a simpler way
than in the proof of \propr{Sglob-fin-gset} and define
$G^\Gamma_n=\setof{x_i}{\dim(x_i)=n}$ (no need to rename variables), which we
will do in the following.

\begin{proposition}
  There is a bijection between pasting schemes~$G$ which are such that the
  sets~$G_n$ are disjoint subsets of the variables and ps-contexts.
\end{proposition}
\begin{proof}
  Suppose given a context such that $\Gamma\vdashps$ holds. We can show by
  induction on its proof that for every sequent $\Gamma\vdashps x:A$, the
  globular set $G^\Gamma$ associated to it is $\triangleleft$-linear and the
  cell corresponding to~$x$ has its iterated targets as only greater elements
  \wrt $\triangleleft$.

  Conversely, suppose given a globular set~$G$ satisfying the hypothesis, with
  cells $G_\infty=\set{x_1,\ldots,x_k}$ such that
  $x_1\triangleleft\ldots\triangleleft x_k$. We can construct by recurrence on
  the length of a prefix of size $i$ (with $1\leq i\leq k$) a derivation of the
  form $\Gamma\vdashps x_i:A_i$ such that $G^\Gamma$ is the subcomplex of~$G$
  generated by $\set{x_1,\ldots,x_i}$ (\ie obtained by taking the closure under
  faces of this set).

  Finally, the two operations can be checked to be mutually inverse.
\end{proof}

\noindent
From this, we finally deduce:

\begin{theorem}
  The category~$\T_0$ is equivalent to the full subcategory~$\Sps$ of~$\Sglob$
  whose objects are pasting schemes.
\end{theorem}


\begin{remark}
  Note that, in fact, we have a tighter correspondence than an equivalence of
  categories since pasting schemes up to isomorphism are in bijection with
  ps-contexts up to $\alpha$\nbd-equivalence (as opposed to isomorphism), \ie
  renaming of variables. Moreover, one can construct a variable-free
  presentation of the sequent calculus (using De Bruijn indices) which entirely
  removes the need for $\alpha$-equivalence.
\end{remark}

\subsection{Boundaries}
We now explain how to compute boundaries (see \defr{ps-boundary}) of
ps-contexts. This is defined as a ``meta-operation'' on contexts.

\begin{definition}
  Given $i\in\N$, we define the \emph{$i$\nbd-source} $\csrc[i](\Gamma)$ of a
  context $\Gamma$ as $\csrc[i](x:\Obj)=x:\Obj$ and
  $\csrc[i](\Gamma,y:A,f:\Hom Axy)$ as
  \[
  \begin{cases}
    \csrc[i](\Gamma)&\text{if $\dim(A)\geq i$}\\
    \csrc[i](\Gamma),y:A,f:\Hom Axy&\text{otherwise}
  \end{cases}
  \]
  and the \emph{$i$-target} by $\ctgt[i](x:\Obj)=x:\Obj$, and
  $\ctgt[i](\Gamma,y:A,f:\Hom Axy)$ as
  \[
  \begin{cases}
    \ctgt[i](\Gamma)&\text{if $\dim(A)>i$}\\
    \textnormal{drop}(\ctgt[i](\Gamma)),y:A&\text{if $\dim(A)=i$}\\
    \ctgt[i](\Gamma)&\text{otherwise}
  \end{cases}
  \]
  where $\textnormal{drop}(\Gamma)$ is $\Gamma$ with the last element
  removed. By convention, we write
  \[
  \csrc(\Gamma)=\csrc[\dim(\Gamma)-1](\Gamma)
  \qquad\qquad
  \ctgt(\Gamma)=\ctgt[\dim(\Gamma)-1](\Gamma)
  \]
\end{definition}

\begin{proposition}
  Given a ps-context $\Gamma$ and $i\in\N$, the contexts $\partial_i^-(\Gamma)$
  and $\partial_i^+(\Gamma)$ are ps-contexts. Moreover, they correspond to
  \defr{ps-boundary}, in the sense that
  \[
  G^{\partial_i^-(\Gamma)}
  \quad\cong\quad
  \partial_i(G^\Gamma)
  \quad\cong\quad
  G^{\partial_i^+(\Gamma)}
  \]
  and the canonical inclusions $G^{\partial_i^-(\Gamma)}\to G^\Gamma$ and
  $G^{\partial_i^+(\Gamma)}\to G^\Gamma$ are respectively $\sigma^G_i$ and
  $\tau^G_i$.
\end{proposition}




\subsection{A non-canonical definition}
To further emphasize the fact that the previous characterization of
ps-contexts is canonical in the sense that each pasting scheme has a
unique derivation, we provide here a second type theoretical
characterization of pasting schemes which lacks this property.

We now consider judgments of the form
\[
\Gamma\vdashps
\qquad\text{and}\qquad
\Gamma\vdashps\Delta
\]
The main modification \wrt the previous definition is that $\Delta$ is a
context, \ie we have a choice of multiple ``free outputs'', whereas we only had
one before. The rules are
\[
\inferrule{\Gamma\vdashps\Delta}{\Gamma\vdashps}
\qquad\qquad
\inferrule{\null}{x:\Obj\vdashps x:\Obj}
\]
\[
\inferrule{\Gamma\vdashps\Delta,x:A,\Delta'}{\Gamma,y:A,f:\Hom Axy\vdashps\Delta,\Delta',y:A,f:A}
\]
where, in the last rule, $y$ and $f$ are fresh in~$\Gamma$.

\begin{proposition}
  For every context $\Gamma$ such that $\Gamma\vdashps$ is derivable $G^\Gamma$
  is a pasting scheme and conversely every pasting scheme is isomorphic to one
  of this form.
\end{proposition}
\begin{proof}
  In a derivation of $\Gamma\vdashps$, we can permute rules so that it
  corresponds (up to bookkeeping) to a derivation in the sense of \secr{ps-tt}.
\end{proof}

\noindent
The following example shows that there is however not a canonical ps-context
associated to a pasting scheme with this variant.

\begin{example}
  The pasting scheme
  \[
  \vxym{
    x\ar@/^2ex/[r]^f\ar@/_2ex/[r]_{f'}\ar@{}[r]|{\Downarrow\alpha}&y\ar[r]^g&z
  }
  \]
  corresponds to both the contexts
  \[
  x:\Obj,y:\Obj,f:\Hom{}xy,f':\Hom{}xy,\alpha:\Hom{}f{f'},z:\Obj,g:\Hom{}yz
  \]
  and
  \[
  x:\Obj,y:\Obj,f:\Hom{}xy,z:\Obj,g:\Hom{}yz,f':\Hom{}xy,\alpha:\Hom{}f{f'}
  \]
  for which we can derive that they are pasting schemes.
\end{example}

From a practical point of view, the main advantage of previous axiomatization of
ps-contexts over this one is that it can be used to simply check whether a
context is a pasting scheme or not, without having to provide a proof or run an
complicated proof-search algorithm.



\section{Weak $\omega$-categories}
\label{sec:cat}
In this section, we finally use our characterization of pasting schemes to give
a type theoretic definition of weak $\omega$\nbd-categories.

\subsection{A type-theoretic definition of $\omega$-groupoids}
The basic idea in order to define an $\omega$-category is that every pasting
scheme should have a composition. We thus introduce a new family of terms to our
syntax, called \emph{coherences}, and denoted
\[
\coh\Gamma A[\sigma]
\]
Each coherence is indexed by a context~$\Gamma$, a type~$A$ and a
substitution~$\sigma$. Such a term should be thought of as a constant
$\coh\Gamma A$ which takes a pasting scheme $\Gamma$ and produces a value of
type~$A$, corresponding to its composition. The substitution~$\sigma$
corresponds to formally applying a substitution to it. By convention, we write
$\coh\Gamma A$ instead of $\coh\Gamma A[\sempty]$ in the following, and we
extend the rules for substitution (see \secr{substitution}) by
\[
(\coh\Gamma A[\sigma])[\tau/\Gamma]
\qeq
\coh\Gamma A[\tau\circ\sigma]
\]
The free variables of a coherence are defined as
\[
\FV(\coh\Gamma A[\sigma])
\qeq
(\FV(A)\setminus\FV(\Gamma))\cup\FV(\sigma)
\]
Indeed, the coherence binds the variables of $\Gamma$ in~$A$. In
practice, we will always have $\FV(A)\subseteq\FV(\Gamma)$ and thus
$\FV(A)\setminus\FV(\Gamma)=\emptyset$.
Note that, contrary to the situation in earlier sections, the addition
of coherences means that terms are no longer necessarily variables,
\ie the distinction between the two syntactic classes is relevant.

Since, in an $\omega$-category, one expects to have composites of
all pasting schemes, a naive definition would simply assert their
existence.  For example, let us consider the consequences of the
following rule:

\begin{equation}
  \label{eq:gpd-rule}
  \inferrule{\Gamma\vdashps\\\Gamma\vdash A}{\Gamma\vdash\coh\Gamma A:A}
\end{equation}

\begin{remark}
  The rule should in fact be written
  \[
  \inferrule{\Gamma\vdashps\\\Gamma\vdash A\\\Delta\vdash\sigma:\Gamma}{\Delta\vdash\coh\Gamma A[\sigma]:\csubst A\sigma\Gamma}
  \]
  so that \lemr{subst-typ} still holds. Alternatively, one can add the rule of
  the lemma to the type theory. We will ignore this detail in what follows.
\end{remark}

\begin{example}
  \label{ex:coh}
  The following coherences are all derivable:
  \begin{itemize}
  \item Every object has an associated identity:
    \[
    x:\Obj\vdash\ccoh:\Hom\Obj xx
    \]
  \item Every pair of composable morphisms have a composition:
    \[
    x:\Obj,y:\Obj,f:\Hom\Obj xy,z:\Obj,g:\Hom\Obj yz\vdash\ccoh:\Hom\Obj xz
    \]
    \newcommand{\ccomp}{\mathsf{comp}}
    \newcommand{\cid}{\mathsf{id}}
  \item There is a morphism witnessing that identities are neutral elements on
    the left (the \emph{left-unitor}):
    \[
    x:\Obj,y:\Obj,f:\Hom\Obj xy\vdash\ccoh:\Hom{\Hom\Obj xy}{\ccomp(\cid_x,f)}f
    \]
    where $\cid_x$ is a notation for the first derived coherence,
    and $\ccomp(\cid_x,f)$ is a notation for the composite of $\cid_x$ and $f$,
    defined using the second coherence. This morphism has a (weak)
    inverse:
    \[
    x:\Obj,y:\Obj,f:\Hom\Obj xy\vdash\ccoh:\Hom{\Hom\Obj xy}f{\ccomp(\cid_x,f)}
    \]
    (this is an inverse only up to weakly invertible morphisms).
  \item Every triple of composable morphisms have a composition:
    \[
    \scalebox{.8}{$x:\Obj,y:\Obj,f:\Hom\Obj xy,z:\Obj,g:\Hom\Obj yz,w:\Obj,h:\Hom\Obj zw\vdash\ccoh:\Hom\Obj xw$}\\
    \]
  \item In fact, this system admits ``partial composition'' operations, ignoring
    some variables in the context.  For
    instance, in a context as above, we can compose only~$f$ and~$g$ (and forget
    about~$h$):
    \[
    \scalebox{.8}{$x:\Obj,y:\Obj,f:\Hom\Obj xy,z:\Obj,g:\Hom\Obj yz,w:\Obj,h:\Hom\Obj zw\vdash\ccoh:\Hom\Obj xz$}\\
    \]
  \item\label{coh-inv} Every morphism admits a weak inverse:
    \[
    x:\Obj,y:\Obj,f:\Hom\Obj xy\vdash\ccoh:\Hom\Obj yx
    \]
  \end{itemize}
\end{example}

\noindent
The last coherence should make it clear that the unrestricted composition
rule above yields not a definition of an $\omega$-category, but of an
\emph{$\omega$-groupoid}. In fact, this definition is very close to Brunerie's
definition~\cite[Appendix~A]{brunerie:these}: this is no accident since our work
is largely inspired by his. The only difference between the two definitions of
$\omega$-groupoid is that Brunerie uses a more liberal notion of pasting scheme,
which is called a \emph{contractible context}, generated by the following
rules:
\[
\inferrule{\null}{x:\Obj\vdashcontr}
\]
\[
\inferrule{\Gamma\vdashcontr\\\Gamma\vdash x:A}{\Gamma,y:A,f:\Hom Axy\vdashcontr}
\qquad\qquad
\inferrule{\Gamma\vdashcontr\\\Gamma\vdash x:A}{\Gamma,y:A,f:\Hom Ayx\vdashcontr}
\]
where $y,f\not\in\FV(\Gamma)$ for the second one. As an illustration of the
difference, the context
\[
x:\Obj,y:\Obj,f:\Hom\Obj xy,y':\Obj,f':\Hom Ax{y'}
\]
is contractible but not a pasting scheme. We will see in next section that using
our pasting schemes allows us to formulate a definition for $\omega$-categories.

\subsection{Type-theoretic definition of $\omega$-categories}
\label{sec:catt}
In order to characterize $\omega$-categories, we will need to restrict
the rule of the previous section in such a way that inverses are
excluded, but such that all reasonable structural operations remain.
In fact, the rule \eqref{eq:gpd-rule} will be replaced by two separate
rules.  Indeed, the ``problem'' with derivation of inverses (last
point of \exr{coh}) is that it exchanges source and target, so we add
a side condition ensuring that this does not happen:
\begin{equation}
  \label{eq:catt-prod}
  \inferrule
  {\scalebox{.9}{$\Gamma\vdashps$}\\\scalebox{.9}{$\Gamma\vdash\Hom Atu$}\\\scalebox{.9}{$\csrc(\Gamma)\vdash t:A$}\\\scalebox{.9}{$\ctgt(\Gamma)\vdash u:A$}}
  {\Gamma\vdash\coh\Gamma{\Hom Atu}:\Hom Atu}
\end{equation}
whenever
\[
\FV(t)=\FV(\csrc(\Gamma))
\qquad\text{and}\qquad
\FV(u)=\FV(\ctgt(\Gamma))
\]
This rule allows one to derive all the ``operations'' required in an
$\omega$-category (\eg composition and identities), but not their coherences
(for instance, the witnesses for identity being a neutral element shown in
previous section are not derivable). We therefore add another rule to compensate:
\begin{equation}
  \label{eq:catt-inv}
  \inferrule
  {\Gamma\vdashps\\\Gamma\vdash A}
  {\Gamma\vdash\coh\Gamma A:A}
\end{equation}
\vspace{-\baselineskip}%
whenever
\[
\FV(A)=\FV(\Gamma)
\]
The side condition forbids ``partial compositions'', which would allow for too
many invertible operations otherwise (see previous section).
%
%
With these two rules, we can derive the operations required to be present in an
$\omega$-category and only those. For instance, all the coherences of \exr{coh}
can be derived excepting the last one. Some more practical illustrations are
given in \secr{implementation}.

Writing $\Scat$ for the syntactic category (of contexts and substitutions)
associated to the preceding type theory, with rules \eqref{eq:catt-prod} and
\eqref{eq:catt-inv} for introducing coherences, we propose the following
definition.

\begin{definition}
  \label{def:tt-cat}
  An \emph{$\omega$-category} is a set-theoretic model of~$\Scat$.
\end{definition}

\subsection{The Grothendieck-Maltsiniotis definition}
In order to motivate our definition on the theoretical side, we briefly recall
the Grothendieck-Maltsiniotis definition of $\omega$-categories in next section
and conjecture that it coincides with our definition,
see~\cite{maltsiniotis2010grothendieck} and \cite{ara:these} for details.

Fix a globular extension~$\Glob\to\C$. A morphism $f$ in~$\C$ is \emph{algebraic} when
for every decomposition $f=g\circ f'$ with~$g$ globular, $g$~is an
identity. From a proof-theoretic perspective, this means that $f$ cannot be
obtained by non-trivially weakening another morphism~$f'$, \ie it ``uses'' all
the cells in its source (requirements below that some morphisms should be
algebraic will give rise to the side conditions of our rules). A pair of
morphisms
\[
f,g
\qcolon
D_i\qto X
\]
is \emph{parallel} when $i>0$ implies
\[
f\circ\sigma_{i-1}=g\circ\sigma_{i-1}
\qquad\text{and}\qquad
f\circ\tau_{i-1}=g\circ\tau_{i-1}
\]
A \emph{lifting} for such a pair is a morphism $h:D_{i+1}\to X$ such that
$f=h\circ\sigma_i$ and $g=h\circ\tau_i$:
\[
\xymatrix{
  D_{i+1}\ar@{.>}[dr]^h\\
  D_i\ar@<.5ex>[u]^{\sigma_i}\ar@<-.5ex>[u]_{\tau_i}\ar@<.5ex>[r]^f\ar@<-.5ex>[r]_g&X
}
\]
A pair of parallel morphisms as above is \emph{admissible} when either
\begin{enumerate}
\item there exist decompositions
  \[
  f=\sigma_{i-1}\circ f'
  \qquad\text{and}\qquad
  g=\tau_{i-1}\circ g'
  \]
  with $f'$ and $g'$ algebraic, or
\item $f$ and~$g$ are algebraic.
\end{enumerate}

The \emph{canonical coherator for $\omega$-categories}, written~$\T$, is the
colimit of the diagram of globular extensions
\[
\vxym{
  \Glob\ar[r]&\T_0\ar[r]&\T_1\ar[r]&\T_2\ar[r]&\cdots
}
\]
where~$\T_{i+1}$ is the globular extension obtained from~$\T_i$ by formally adding
a lifting for every admissible pair of arrows, and taking the globular extension
freely generated by the resulting category (see \secr{gext}).

\begin{definition}[\cite{maltsiniotis2010grothendieck}]
  \label{def:cat}
  An \emph{$\omega$-category}~$C$ is a functor $C:\T^\op\to\Set$ such that
  $C^\op$ preserves globular sums.
\end{definition}

\noindent
The two rules \eqref{eq:catt-prod} and
\eqref{eq:catt-inv} for introducing coherences in \secr{catt}
correspond precisely to formally adding lifting for admissible
morphisms, with each rule corresponding to one of the conditions for
being admissible. The details, however, are rather involved and left
for future work.

\begin{conjecture}
  The category $\Scat$ is equivalent to $\T$.
\end{conjecture}

\newcommand{\cd}{\operatorname{cd}}

\noindent
More precisely, we define the \emph{coherence depth}~$\cd(t)$ of a term~$t$ as
the number of nested coherences, \ie $\cd(x)=0$,
$\cd(\coh\Gamma A\sigma)=\max(\cd(A)+1,\cd(\sigma))$, etc. Given $n\in\N$, we
conjecture that the subcategory of $\Scat$, with the same objects, morphisms
being substitutions with coherence depth less than~$n$, is equivalent to
$\T_n$. Finally, we conjecture that the situation \wrt set-theoretic models
described \secr{models} generalizes as follows.

\begin{conjecture}
  Type-theoretic $\omega$-categories (\defr{tt-cat}) correspond precisely to
  Grothendieck-Maltsiniotis $\omega$\nbd-cate\-go\-ries (\defr{cat}).
\end{conjecture}





\section{Implementation(s)}
\label{sec:implementation}
Since there are two authors for this paper, there are also two implementations
of a type-checker for the theory.
The first\footnote{\url{https://github.com/ericfinster/catt}} is
done in Haskell, following precisely the inference rules described in this
article, while the other\footnote{\url{https://github.com/smimram/catt}} is in
OCaml and has some more experimental features (notably, the presence of implicit
arguments making proofs much shorter but lacking theoretical
justification for the moment).
The second implementation may be tried
online\footnote{\url{https://smimram.github.io/catt}}.

In fact, our definition naturally lends itself to standard techniques
for the implementation of a type checking algorithm for a dependent
type theory, although in view of sparseness of the theory, these
techniques appear in a rather simplified form. The Haskell
implementation, for example, uses a simple bi-directional typechecking
setup together with normalization by evaluation. Furthermore, while the
theory lacks any notion of abstraction, the coherences are nonetheless
assigned the type of a dependent product internally, allowing for
substitution to propagate in the types. Indeed one sees immediately by
inspection of the rules that each coherence can naturally be seen as a
formal constant in the dependent product obtained by abstracting
over all of the variables in the context (necessarily a pasting
scheme) which defines it.

In our system, the user writes statements of the form
\[
x_1:A_1,\ldots,x_n:A_n\vdash \ccoh:A
\]
and the typechecker automatically ensures that the judgment is derivable, or issues
an error if it is not the case. In practice, coherences are written
\[
\texttt{coh $name$ ($x_1$ : $A_1$) $\ldots$ ($x_n$ : $A_n$) : $A$ ;}
\]
where $name$ allows the user to give a name to a coherence. Note that the order in
which arguments are given is important, as it is used to determine whether
the corresponding context is a ps-context or not. The arrow type $\Hom Axy$ is
noted
\[
\texttt{$A$ | $x$ -> $y$}
\]
For instance, we can define identities on $0$-cells:
\begin{verbatim}
coh id (x : *) : * | x -> x ;
\end{verbatim}
composition of $1$-cells:
\begin{verbatim}
coh comp (x : *) (y : *) (f : * | x -> y)
         (z : *) (g : * | y -> z)
         : * | x -> z ;
\end{verbatim}
left unitor:
\begin{verbatim}
coh unitl (x : *) (y : *) (f : * | x -> y)
          : * | x -> y
              | comp x x (id x) y f -> f ;
\end{verbatim}
the ``inverse'' for left unitor:
\begin{verbatim}
coh unitl' (x : *) (y : *) (f : * | x -> y)
           : * | x -> y
               | f -> comp x x (id x) y f ;
\end{verbatim}
%
%
%
%
associativity of composition of $1$-cells:
\begin{verbatim}
coh assoc
    (x : *) (y : *) (f : * | x -> y) (z : *)
    (g : * | y -> z) (w : *) (h : * | z -> w)
    : * | x -> w
        | comp x z (comp x y f z g) w h ->
          comp x y f w (comp y z g w h) ;
\end{verbatim}
vertical composition of $2$-cells:
\begin{verbatim}
coh vcomp
    (x : *) (y : *) (f : * | x -> y)
    (g : * | x -> y) (a : * | x -> y | f -> g)
    (h : * | x -> y) (b : * | x -> y | g -> h)
    : * | x -> y | f -> h ;
\end{verbatim}
horizontal composition of $2$-cells:
\begin{verbatim}
coh hcomp
  (x : *) (y : *) (f : * | x -> y)
  (g : * | x -> y) (a : * | x -> y | f -> g)
  (z : *) (h : * | y -> z) (k : * | y -> z)
  (b : * | y -> z | h -> k)
  : * | x -> z
      | comp x y f z h -> comp x y g z k ;
\end{verbatim}
the exchange law:
\begin{verbatim}
coh ichg
  (x : *) (y : *) (f : * | x -> y)
  (g : * | x -> y) (a : * | x -> y | f -> g)
  (h : * | x -> y) (b : * | x -> y | g -> h)
  (z : *) (l : * | y -> z) (m : * | y -> z)
  (c : * | y -> z | l -> m) (n : * | y -> z)
  (d : * | y -> z | m -> n)
  : * | x -> z
  | comp x y f z l -> comp x y h z n
  | hcomp x y f h (vcomp x y f g a h b) z l n
    (vcomp y z l m c n d) ->
    vcomp x z (comp x y f z l) (comp x y g z m)
    (hcomp x y f g a z l m c) (comp x y h z n)
    (hcomp x y g h b z m n d) ;
\end{verbatim}
Finally, and as expected, defining an ``inverse'' for an arbitrary $1$-cell fails: the input
\begin{verbatim}
coh inv (x : *) (y : *) (f : * | x -> y)
        : * | y -> x ;
\end{verbatim}
produces the following output:
\begin{verbatim}
Checking coherence: inv
Valid tree context
Src/Tgt check forced
Source context: (x : *)
Target context: (y : *)
Failure: Source is not algebraic for y : *
\end{verbatim}
meaning that the side conditions of the rule \eqref{eq:catt-prod} are not
fulfilled.


\section{Conclusion and future work}
\label{sec:concl}
We have presented a type theory designed to capture a well-known
definition of $\omega$-category, extending work on a similar
definition for $\omega$-groupoids.  Most importantly, we have examined
the relationship between pasting schemes represented as well-formed
contexts, and their semantic counterparts (Batanin trees and globular sums).
We conjecture that the models of this theory coincide with the
definition of Maltsiniotis, but a detailed comparison will have to
await further work.

We note also that the combinatorics of pasting schemes, as described by their
$\triangleleft$-relation, seems promising. It quickly reminds one of
the theory of Dyck words and we expect that interesting results can be
obtained by applying similar methods.

Brunerie's definition of $\omega$-groupoids, upon which this work
builds, was of course motivated by the view of types advocated in
homotopy type theory. Since the introduction of homotopy type theory~\cite{hottbook},
many authors have wondered about the possibility of weakening the
equality relation in order to obtain a theory in which types behave as
\emph{categories} or \emph{directed} homotopy types.  We feel that the
theory presented in this paper serves as a small step in this direction,
isolating the core system of coherences which one would like to have.
In future work, we aim to see if other type theoretic constructions
($\Sigma$ and $\Pi$ types, for example) may be reasonably added to the
theory, thus increasing its expressive power.






\bibliographystyle{IEEEtranS}
\bibliography{papers}

\begin{thebibliography}{10}
\providecommand{\url}[1]{#1}
\csname url@samestyle\endcsname
\providecommand{\newblock}{\relax}
\providecommand{\bibinfo}[2]{#2}
\providecommand{\BIBentrySTDinterwordspacing}{\spaceskip=0pt\relax}
\providecommand{\BIBentryALTinterwordstretchfactor}{4}
\providecommand{\BIBentryALTinterwordspacing}{\spaceskip=\fontdimen2\font plus
\BIBentryALTinterwordstretchfactor\fontdimen3\font minus
  \fontdimen4\font\relax}
\providecommand{\BIBforeignlanguage}[2]{{%
\expandafter\ifx\csname l@#1\endcsname\relax
\typeout{** WARNING: IEEEtranS.bst: No hyphenation pattern has been}%
\typeout{** loaded for the language `#1'. Using the pattern for}%
\typeout{** the default language instead.}%
\else
\language=\csname l@#1\endcsname
\fi
#2}}
\providecommand{\BIBdecl}{\relax}
\BIBdecl

\bibitem{altenkirch2012syntactical}
T.~Altenkirch and O.~Rypacek, ``{A Syntactical Approach to Weak
  $\omega$\nbd-Groupoids},'' in \emph{LIPIcs-Leibniz International Proceedings
  in Informatics}, vol.~16.\hskip 1em plus 0.5em minus 0.4em\relax Schloss
  Dagstuhl-Leibniz-Zentrum fuer Informatik, 2012.

\bibitem{ara:these}
D.~Ara, ``{Sur les $\infty$-groupoïdes de Grothendieck et une variante
  $\infty$\nbd-catégorique},'' Ph.D. dissertation, Université Paris Diderot,
  2010.

\bibitem{batanin1998monoidal}
M.~A. Batanin, ``{Monoidal Globular Categories As a Natural Environment for the
  Theory of Weak $n$-Categories},'' \emph{Advances in Mathematics}, vol. 136,
  no.~1, pp. 39--103, 1998.

\bibitem{berger2002cellular}
C.~Berger, ``A cellular nerve for higher categories,'' \emph{Advances in
  Mathematics}, vol. 169, no.~1, pp. 118--175, 2002.

\bibitem{brunerie:these}
G.~Brunerie, ``{On the homotopy groups of spheres in homotopy type theory},''
  Ph.D. dissertation, {Universit{\'e} Nice Sophia Antipolis}, Jun. 2016.

\bibitem{cartmell1986generalised}
J.~Cartmell, ``Generalised algebraic theories and contextual categories,''
  \emph{Annals of Pure and Applied Logic}, vol.~32, pp. 209--243, 1986.

\bibitem{dybjer1995internal}
P.~Dybjer, ``Internal type theory,'' in \emph{International Workshop on Types
  for Proofs and Programs}.\hskip 1em plus 0.5em minus 0.4em\relax Springer,
  1995, pp. 120--134.

\bibitem{opetopic}
E.~Finster, ``Opetopic,'' \url{http://opetopic.net/}.

\bibitem{grothendieck1983pursuing}
A.~Grothendieck, ``Pursuing stacks,'' 1983, unpublished manuscript.

\bibitem{joyal1997disks}
A.~Joyal, ``Disks, duality and {$\Theta$}-categories,'' \emph{Preprint}, 1997.

\bibitem{lumsdaine2009weak}
P.~L. Lumsdaine, ``Weak $\omega$-categories from intensional type theory,'' in
  \emph{International Conference on Typed Lambda Calculi and
  Applications}.\hskip 1em plus 0.5em minus 0.4em\relax Springer, 2009, pp.
  172--187.

\bibitem{maclane2012sheaves}
S.~Mac~Lane and I.~Moerdijk, \emph{{Sheaves in geometry and logic: A first
  introduction to topos theory}}.\hskip 1em plus 0.5em minus 0.4em\relax
  Springer Science \& Business Media, 2012.

\bibitem{maltsiniotis2010grothendieck}
G.~Maltsiniotis, ``Grothendieck $\infty$-groupoids, and still another
  definition of $\infty$-categories,'' \emph{arXiv preprint arXiv:1009.2331},
  2010.

\bibitem{pare1990simply}
R.~Par{\'e}, ``Simply connected limits,'' \emph{Canadian Journal of
  Mathematics}, vol.~42, no.~4, pp. 731--746, 1990.

\bibitem{street1991parity}
R.~Street, ``Parity complexes,'' \emph{Cahiers de topologie et
  g{\'e}om{\'e}trie diff{\'e}rentielle cat{\'e}goriques}, vol.~32, no.~4, pp.
  315--343, 1991.

\bibitem{hottbook}
T.~{Univalent Foundations Program}, \emph{Homotopy Type Theory: Univalent
  Foundations of Mathematics}.\hskip 1em plus 0.5em minus 0.4em\relax Institute
  for Advanced Study: \url{https://homotopytypetheory.org/book}, 2013.

\bibitem{van2011types}
B.~Van Den~Berg and R.~Garner, ``Types are weak $\omega$-groupoids,''
  \emph{Proceedings of the London Mathematical Society}, vol. 102, no.~2, pp.
  370--394, 2011.

\bibitem{vicary2016globular}
J.~Vicary, A.~Kissinger, and K.~Bar, ``Globular: an online proof assistant for
  higher-dimensional rewriting,'' in \emph{1st International Conference on
  Formal Structures for Computation and Deduction (FSCD 2016)}, vol.~52, 2016,
  pp. 1--11.

\end{thebibliography}
\end{document}